\RequirePackage{fix-cm}
\documentclass[twocolumn]{svjour3}

\def\makeheadbox{\relax}

\makeatletter
\def\@seccntformat#1{\@ifundefined{#1@cntformat}%
  {\csname the#1\endcsname\quad}  % default
  {\csname #1@cntformat\endcsname}% enable individual control
}
\let\oldappendix\appendix %% save current definition of \appendix
\renewcommand\appendix{%
    \oldappendix
    \newcommand{\section@cntformat}{\appendixname~\thesection\quad}
}
\makeatother
\smartqed  % flush right qed marks, e.g. at end of proof

\usepackage{graphicx,bbm,setspace,amssymb,amsfonts,amsmath,mathtools,mathrsfs,stmaryrd,bm,bm,etoolbox,comment,hyperref,url,color,enumitem,algorithm,booktabs,microtype,nicefrac,lingmacros,nccmath,tree-dvips,tikz-cd,xurl,soul}

\usepackage[colorinlistoftodos]{todonotes}
\usepackage[utf8]{inputenc} % allow utf-8 input
\usepackage[T1]{fontenc}    % use 8-bit T1 fonts
\DeclareMathOperator*{\loggrad}{log-grad}

\usepackage[noend]{algpseudocode}

\usepackage{graphicx}
\usepackage{subcaption}
\usepackage{appendix}

\journalname{Nonlinear Dynamics}
\begin{document}

\title{Collective correlations, dynamics, and behavioural inconsistencies of the cryptocurrency market over time}
%\titlerunning{Short form of title}% if too long for running head

\author{Nick James \and Max Menzies}

%\authorrunning{Short form of author list} % if too long for running head

\institute{   
                N. James \at
              School of Mathematics and Statistics\\
              University of Melbourne\\
              VIC, 3010, Australia\\
              \email{nick.james@unimelb.edu.au} 
              \and
                M. Menzies \at
              Beijing Institute of Mathematical Sciences and Applications\\
              Tsinghua University\\
              Beijing, 101408, China
           }

\date{Received: date / Accepted: date}
% The correct dates will be entered by the editor

\maketitle

\begin{abstract}

This paper introduces new methods to study behaviours among the 52 largest cryptocurrencies between 01-01-2019 and 30-06-2021. First, we explore evolutionary correlation behaviours and apply a recently proposed turning point algorithm to identify regimes in market correlation. Next, we inspect the relationship between collective dynamics and the cryptocurrency market size - revealing an inverse relationship between the size of the market and the strength of collective dynamics. We then explore the time-varying consistency of the relationships between cryptocurrencies' size and their returns and volatility. There, we demonstrate that there is greater consistency between size and volatility than size and returns. Finally, we study the spread of volatility behaviours across the market changing with time by examining the structure of Wasserstein distances between probability density functions of rolling volatility. We demonstrate a new phenomenon of increased uniformity in volatility during market crashes, which we term \emph{volatility dispersion}.

\keywords{Cryptocurrency \and Time series analysis \and Nonlinear dynamics \and  Correlations \and Collective behaviours}

\end{abstract}

%---------------------------- BODY OF PAPER---------------

\section{Introduction}
\label{Introduction}

%1. Overview of cryptocurrencies - set stage for paper
Over the last several years, the cryptocurrency market has attracted substantial interest from both institutional and retail investors. The market has experienced significant growth in total assets, accompanied by commensurate levels of volatility. Following the COVID-19 and BitMEX market crash in early 2020, cryptocurrency prices had a strong rally and then a subsequent decline in prices across the market. Given the differing views on their long-term viability, understanding the changing nature of cryptocurrencies' returns and volatility as the market grows in size is a timely priority for investors and policymakers alike. Our paper meets this purpose by analysing the collective dynamics of cryptocurrencies over time regarding returns, volatility, and market size.

%2. Existing research on finance in general
Our paper builds on a long literature studying the dynamics of financial markets. %Within the nonlinear dynamics and econophysics communities,
Across several fields, researchers have been interested in the time-varying nature of financial market behaviours for some time, particularly correlations \cite{Fenn2011,Mnnix2012,Vicente2006,Wang2013,GangJinWang2012,Mikiewicz2021,Pan2007}. The application of classical statistical techniques, such as ARCH and GARCH models \cite{Lamoureux1990,Chu2017,Kumar2019} and other parametric models, can encounter difficulties due to the non-stationary nature of financial markets. Certain models for descriptive analysis or portfolio selection that perform well during extended bull market periods can suffer during sudden market crises. The literature has split into two approaches to this problem. Statistical researchers have introduced explicit methodologies to model non-stationarity \cite{Dahlhaus1997}, while nonlinear dynamics researchers have taken a more descriptive approach to the time-varying dynamics.

%2.5
In the nonlinear dynamics community, such market dynamics have been studied with a range of methodologies, including chaotic systems \cite{Cai2012,Tacha2018,Szumiski2018}, clustering \cite{Heckens2020,Jamesfincovid}, sample entropy \cite{Wu2021,Chen2021}, and principal components analysis \cite{Laloux1999,Kim2005,JamescryptoEq}. Various asset classes have attracted interest, including equities \cite{Wilcox2007}, fixed income \cite{Driessen2003}, and foreign exchange \cite{Ausloos2000}. Researchers have also explored extreme behaviours \cite{Qi2019} and the propagation of structural breaks \cite{Telli2020,James2021_crypto} in price and volatility time series. The evolutionary nature of volatility, often studied via the framework of volatility clustering and regimes, has been of interest for many years \cite{Shah2019,Kirchler2007,Baillie2009,Hamilton1989,Lavielle,arjun,Lamoureux1990,Guidolin2007,Yang2018JEDC,deZeeuw2012,Kumar2019}. Overall, such financial research uses many of the same techniques from time series analysis that are used in other domains \cite{Hethcote2000,James2021_virulence,Vazquez2006,Mendes2018,Mendes2019,Rizzi2010,Shang2020,jamescovideu,Machado2020,James2021_geodesicWasserstein}.

%3. Crypto research
In recent years, substantial research has focused specifically on the unique dynamics of cryptocurrencies. Various topics of interest include Bitcoin and other cryptocurrencies' price dynamics \cite{Chu2015,Lahmiri2018,Kondor2014,Bariviera2017,AlvarezRamirez2018}, fractal patterns, \cite{Stosic2019,Stosic2019_2,Manavi2020,Ferreira2020}, cross-correlation and scaling effects \cite{Drod2018,Drod2019,Drod2020,Gbarowski2019,Drod2021_entropy,Wtorek2021_entropy}. A great deal of this work has addressed how these dynamics have changed over time, in particular during market crises such as the COVID-19 market crash \cite{Wtorek2020,Corbet2020,Conlon2020,Conlon2020_2,Ji2020,Lahmiri2020,Zhang2020finance,He2020,Zaremba2020,Akhtaruzzaman2020,Okorie2020,Naeem2021,Curto2021,james2021_mobility}.

%4. non-stationarity
This paper prioritises understanding the changing nature of time-varying parameters that describe the market dynamics of cryptocurrencies. We are particularly interested in a descriptive analysis of these non-stationary dynamics and an identification of different patterns of behaviour, particularly around crises. Our primary parameters of interest are cryptocurrency returns $R(t)$, volatility $\Sigma(t)$, which have been well-studied and known to be highly non-stationary, as well as market size $M(t)$. This latter quantity is particularly relevant to a nascent market such as cryptocurrencies to determine if market dynamics change meaningfully as the market develops.

%Structure
This work aims to study the changing dynamics of the cryptocurrency market, incorporating and building on much of the aforementioned literature. We take an interest both across our entire 2.5-year window of analysis as well as in several distinct periods, which are outlined in Section \ref{data}. In Section \ref{Market_correlations}, we begin by building on the rich literature studying cryptocurrencies' significant correlations, proposing a new framework to analyse the correlation structure of the market. In Section \ref{Collective_dynamics_fragmentation}, we investigate the relationship between the collective dynamics of the market and the time-varying total market capitalisation of all cryptocurrencies. From there in Section \ref{return_vol_size_inconsistency}, we study the relationship between market size and returns and volatility, demonstrating greater quantitative consistency between market size and volatility than size and returns. Finally, Section \ref{Volatility_persistence_regime_identification} introduces a new technique for understanding the changing spread of volatility across the cryptocurrency market as a whole, terming this \emph{volatility dispersion}. There we introduce a new quantity  $\text{Var}(\mathbf{p}(t))$ to describe the spread of volatility across a market that we also observe is highly non-stationary. Overall, we propose a new methodology to understand the evolution of correlation structures, reveal new insights about the relationship between market size and attributes such as collective dynamics, returns and volatilities, and propose a new way to understanding the spread of volatility with time. Our insights are summarised in Section \ref{Conclusion}.

\section{Data}
\label{data}
In the proceeding sections, we analyse cryptocurrency data between 01-01-2019 and 30-06-2021. We study 52 cryptocurrencies that possess sufficient histories. In Section \ref{Market_correlations}, we partition our analysis into five discrete periods to explore correlation behaviours at varying times. These periods are defined as follows:
\begin{enumerate}
    \item Pre-COVID: 01-01-2019 to 28-02-2020;
    \item Peak COVID: 01-03-2020 to 30-05-2020;
    \item Post-COVID: 31-05-2020 to 31-08-2020;
    \item Bull: 01-09-2020 to 14-04-2021;
    \item Bear: 15-04-2021 to 30-06-2021. 
\end{enumerate}
Cryptocurrency data are sourced from \url{https://coinmarketcap.com/}. A full list of cryptocurrencies studied in this paper is available in Appendix \ref{appendix:mathematical_objects}.

\section{Temporal evolution of market correlation}
\label{Market_correlations}

Like many asset classes in financial markets, cryptocurrency returns are characterised by distributions that exhibit significant tail risk. In particular, their relative infancy and polarising views of the asset's long-term viability make their price behaviour more susceptible to extreme volatility and erratic behaviours \cite{James2021_crypto}. Figure \ref{fig:market_log_returns} displays the general volatility in log returns, and Figure \ref{fig:market_log_returns_distribution} highlights the negative skew in the returns distribution throughout our analysis window. 

\begin{figure*}
    \centering
    \begin{subfigure}[b]{0.48\textwidth}
        \includegraphics[width=\textwidth]{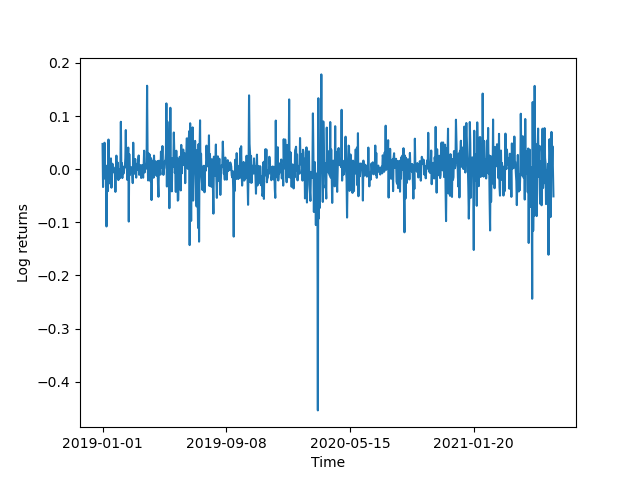}
        \caption{}
    \label{fig:market_log_returns}
    \end{subfigure}
    \begin{subfigure}[b]{0.48\textwidth}
        \includegraphics[width=\textwidth]{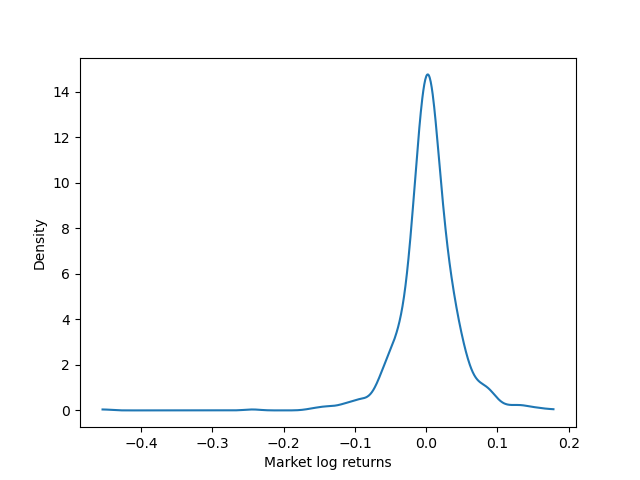}
        \caption{}
    \label{fig:market_log_returns_distribution}
    \end{subfigure}
    \caption{Cryptocurrency market log returns, depicted (a) as a function of time and (b) as a distribution.}
    \label{fig:Market_returns_distribution}
\end{figure*}

Prior research has demonstrated that correlations are most strongly positive during market crises, where many investors engage in the systematic sale of assets \cite{Sandoval2012,james2021_MJW}. This may be due to the growth of quantitatively driven asset managers and their use of algorithms that may induce simultaneous indiscriminate selling. That is, market systematic behaviours, particularly during bear markets, appear to be more pronounced due to an \emph{algorithmic herd mentality}. Behaviour has been particularly volatile since 2020, where after the COVID-19 and BitMEX crash, the cryptocurrency market experienced an unprecedented rise (and subsequent fall) in total market capitalisation. With all this in mind, our objective is to study the temporal evolution of the market's correlation structure and contrast the market's collective similarity between different periods.

\subsection{Analysis of collective correlation over time}
\label{sec:evolutionarycorrelation}
Let our period of analysis 01-01-2019 to 30-06-2021 be indexed $t=0,1,...,T$, where $T=911$. Let $c_i(t), i=1,...,N, t=0,...,T$ be the multivariate time series of cryptocurrency daily closing prices. We first generate a multivariate time series of log returns, $R_i(t), t=1,...,T$, as follows:
\begin{align}
\label{eq:logreturns}
R_{i}{(t)} &= \log \left(\frac{c_i{(t)}}{c_i{(t-1)}}\right).
\end{align}
Our primary objects of study in this section are correlation matrices of log returns data across specified periods. These periods may roll forward with time or remain static. Let $a \leq t \leq b$ be such a period of analysis, an interval of $S=b-a + 1$ days. We standardise the cryptocurrency log returns over such a period by defining $\tilde{R}_i(t) = [R_i(t) - \langle R_i \rangle] / \sigma(R_i), a \leq t \leq b$, where $\langle . \rangle $ is an average and $\sigma(.)$ is a standard deviation operator, each computed over the same interval. The correlation matrix $\Psi$ is then defined as follows: let $\tilde{R}$ be a $N \times S$ matrix defined by $\tilde{R}_{it}=\tilde{R}_i(t), i=1,...,N, t=a,...,b$ and define
\begin{align}
\label{eq:corrmatrix}
\Psi = \frac{1}{S} \tilde{R} \tilde{R}^T.
\end{align}
Explicitly, individual entries are defined by
\begin{align}
\label{eq:rhodefn}
    \Psi_{ij}=\frac{\sum_{t=a}^b (R_i(t) - \langle R_i \rangle)(R_j(t) - \langle{R}_j \rangle))}{\left(\sum_{t=a}^b (R_i(t) - \langle R_i \rangle)^2 \sum_{t=a}^b (R_j(t) - \langle R_j \rangle)^2\right)^{1/2}},
\end{align}
for  $1\leq i, j\leq N$. All entries $\Psi_{ij}$ lie in  $[-1,1]$. If we wish to explicitly note the interval over which these are defined, we may denote this matrix $\Psi^{[a:b]}.$ To quantify the total strength in correlation behaviours across the market, we compute an appropriately normalised $L^1$ norm of the matrix $\Psi$. That is, let
\begin{align}
\label{eq:L1norm}
    \|\Psi\|_1 = \frac{1}{N^2} \sum_{i,j=1}^N | \Psi_{ij}|.
\end{align}
This gives the average absolute correlation of all cryptocurrencies over the interval $a \leq t \leq b$. To explore the temporal evolution of collective strength in correlation behaviours, we examine the changes in matrix $\Psi$ as our interval $[a,b]$ rolls forward. Specifically, we set $S=90$ and compute the time-varying $L^1$ norm of a $90$-day rolling window, $\nu^{\Psi}(t) = \|\Psi^{[t-S+1:t]}\|_1$. We also apply a \emph{Savitzky-Golay} filter to produce a smoothed function $\nu^{\Psi}(t)$. We then apply a recently introduced turning point algorithm \cite{james2020covidusa}, detailed in Appendix \ref{appendix:TPA}, to generate a set of non-trivial local maxima and minima in the total correlation behaviours. While some previous work explores the evolution of the first eigenvalue \cite{JamescryptoEq}, our methodology is the first we know to study the $L^1$ norm of the correlation matrix and apply a bespoke turning point algorithm to study non-trivial peak and trough propagation. The norm function $\nu^{\Psi}(t)$, its smoothed analogue $\nu_s^{\Psi}(t)$, and the detected turning points are all displayed in Figure \ref{fig:Correlation_matrix_norm}.

\begin{figure*}
    \centering
    \includegraphics[width=0.9\textwidth]{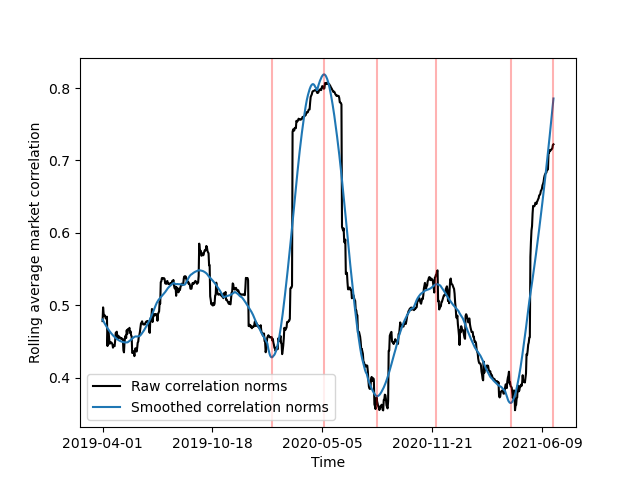}
    \caption{
    Time-varying correlation matrix norms, $\nu^{\Psi}(t)$ and its smoothed counterpart $\nu_{s}^{\Psi}(t)$, as defined in Section \ref{sec:evolutionarycorrelation}. Six non-trivial local maxima and minima are annotated, detected at 03-02-2020, 08-05-2020, 13-08-2020, 29-11-2020, 13-04-2020 and 30-06-2020.
    }
    \label{fig:Correlation_matrix_norm}
\end{figure*}

Examining Figure \ref{fig:Correlation_matrix_norm} reveals six non-trivial local maxima and minima in the overall magnitude of correlations. Of particular note are the local maximum identified on 08-05-2020 and the minimum on 13-08-2020. These dates reflect the COVID-19 market crash and the subsequent recovery in the cryptocurrency market, respectively. Clearly, correlations are most significant during the crash and weakest during the subsequent growth in the market. The local maximum on 30-06-2020 reflects a significant increase in correlation behaviours during the latter part of our analysis window. This is most likely indicative of the aggressive sell-off in cryptocurrency assets from approximately mid-April until the present. Again, we see a striking pattern where the general growth of the cryptocurrency market is inversely related to the collective strength of correlations.

This is broadly consistent with general economic intuition. There is a substantial amount of literature in the quantitative finance and applied mathematics communities that highlights that an increase in correlation among financial securities is observed during times of crisis \cite{JamescryptoEq}. This corresponds to spikes in the first eigenvalue in the correlation matrix - indicating increases in collective market behaviour. When money flows out of the cryptocurrency market, through losses and sale of assets, this would lead to increased correlations and a spike in the first eigenvalue. Thus, the inverse relationship we identify is consistent with what one would expect based on prior research \cite{Fenn2011}.

However, our findings go further and may have fruitful implications for investment managers. Identifying such peaks and troughs in correlation behaviours may provide a decision support tool as to whether they are in a market environment where security selection is of greater or less importance. In market environments where correlations are strongly positive, it may be more difficult for managers to produce return streams that exhibit lower market beta. The approximate periodicity we observe in local maxima and minima may encourage cryptocurrency investors to engage in cyclical patterns of more bullish and bearish investing to avoid exposure in riskier periods.

\subsection{Comparison of particular periods}
\label{sec:intervalscorrelation}

To further elucidate the findings of the previous section, we partition our analysis window into five periods and analyse the correlations separately within each of these windows. Let $\Psi^{\text{Pre}}$, $\Psi^{\text{Peak}}$, $\Psi^{\text{Post}}$, $\Psi^{\text{Bull}}$ be the correlation matrices obtained across the entire time intervals specified in Section \ref{data}. In the notation of Section \ref{sec:evolutionarycorrelation}, $\Psi^{\text{Pre}}=\Psi^{[1:59]}$, and so on.

\begin{figure*}
    \centering
    \includegraphics[width=\textwidth]{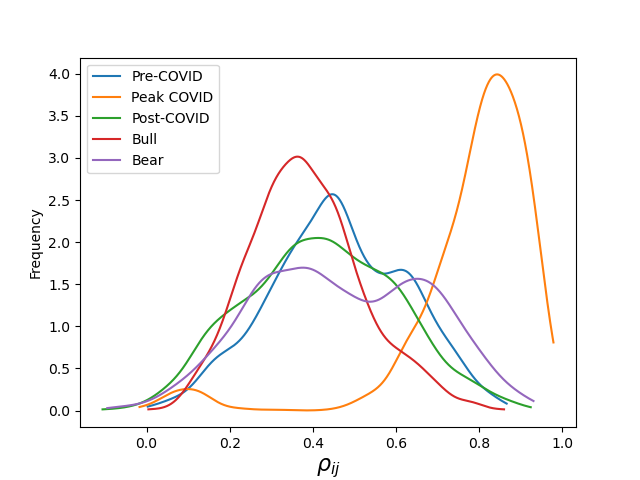}
    \caption{
    Kernel density estimates of correlation matrix entries corresponding to our five time partitions: Pre-COVID, Peak COVID, Post-COVID, Bull and Bear. The Peak COVID correlation matrix has the highest average correlation, indicating extraordinarily high correlations during this period.
    }
    \label{fig:Correlation_matrix_distributions}
\end{figure*}

Figure \ref{fig:Correlation_matrix_distributions} displays kernel density estimates of the entries of our 5 correlation matrices. The findings are generally consistent with Figure \ref{fig:Correlation_matrix_norm}. The distribution of elements $\Psi_{ij}$ in the four matrices $\Psi^{\text{Pre}}$, $\Psi^{\text{Post}}$, $\Psi^{\text{Bull}}$ and $\Psi^{\text{Bear}}$ are highly similar. All four distributions have means in the range $0.38$ to $0.46$ and exhibit similar variability around their means. By contrast, $\Psi^{\text{Peak}}$ has a mean value of $0.78$ - highlighting the spike in correlations during the COVID-19 crisis. One finding of slight surprise is the Bear partition exhibiting correlations more like the Pre-COVID, Post-COVID and Bull windows than the Peak COVID window. It is quite possible that if the dates of the Bear period were brought closer to the present day, including beyond 30-06-2021, this distribution would yield a significantly higher average correlation score. The means and standard deviations of the entries of all five matrices are recorded in Table \ref{tab:table_correlation_matrix}.

The observed increase in correlations during market crises is consistent with prior research that has been published in the literature. However, one new observation is the stark contrast in correlation behaviours in periods directly adjacent to an acute period of crisis. One could imagine that the distribution of cryptocurrency correlations transitions gradually downward following a market crisis. Instead, periods directly after such crises exhibit correlations that are more similar to relatively stable market conditions.

\begin{table}[ht]
\centering
\begin{tabular}{ |p{2cm}|p{2.8cm}|p{2.9cm}|}
 %\hline
% \multicolumn{3}{|c|}{Correlation matrix distribution moments} \\
 \hline
 Period & Mean of entries $\Psi_{ij}$ & Standard deviation  \\
 \hline
Pre-COVID & 0.456 & 0.164 \\
Peak COVID & 0.784 & 0.166 \\
Post-COVID & 0.421 & 0.182 \\
Bull & 0.383 & 0.135 \\
Bear & 0.421 & 0.182 \\
\hline
\end{tabular}
\caption{Mean and standard deviation of correlation matrix entries across the five periods, as analysed in Section \ref{sec:intervalscorrelation}.}
\label{tab:table_correlation_matrix}
\end{table}

\section{Collective dynamics and market size}
\label{Collective_dynamics_fragmentation} 
In the previous section, we observed weaker correlation behaviours while the market grew and stronger correlation behaviours when the market was in crisis or contracting. In this section, we take a closer examination of the underlying market effect of collective dynamics and how this relates to the current market size.

First, we apply principal components analysis (PCA) to our time-varying correlation matrices. This procedure learns the linear map $\Omega$ such that our standardised returns matrix $\tilde{R}$ is transformed into a matrix of uncorrelated variables $Z$, that is, $Z = \Omega \tilde{R}$. The rows of $Z$ represent the principal components (PCs) of the matrix $\tilde{R}$, while the rows of $\Omega$ consist of the principal component coefficients. The matrix is ordered such that the first row is along the axis of most variation in the data. Its corresponding eigenvalue $\lambda_1$ is thus of substantial practical importance, quantifying the greatest extent of variance in the data. It has been referred to as representative of the collective strength of the market \cite{Fenn2011}. All subsequent PCs, subject to the constraint that they are mutually orthogonal, maximise the variance along their respective axes. Continuing this procedure effectively diagonalises the correlation matrix, $\Psi = E D E^{T}$, where $D$ is a diagonal matrix of eigenvalues and $E$ is an orthogonal matrix. By (\ref{eq:rhodefn}), $\Psi$ is a symmetric positive semi-definite matrix with all eigenvalues real and non-negative, so we may order them $\lambda_1 \geq ... \geq \lambda_N \geq 0$. Each $\lambda_i$ quantifies the extent of variance along the $i$th principal component axis. Thus, we may normalise the eigenvalues by defining $\tilde{\lambda}_i = \frac{\lambda_i}{\sum^{N}_{j=1} \lambda_j}$ to determine the proportion of all variance accounted for by each step in the PCA.

This quantity is related to the norm of the correlation matrix defined in Section \ref{sec:evolutionarycorrelation}. Indeed, by the spectral theorem, the largest (absolute value) eigenvalue of a symmetric matrix coincides with the matrix' \emph{operator norm} \cite{RudinFA}. That is, 
\begin{align}
|\lambda_1|=\|\Psi\|_{op}=\max_{x \in \mathbb{R}^N - \{0\}} \frac{\|\Psi x\|}{\|x\|}.
\end{align}
Next, every diagonal entry of $\Psi$ is equal to 1, so the trace of $\Psi$ is equal to $N$, and thus $\sum^{N}_{j=1} \lambda_j=N$. Hence $\tilde{\lambda}_1 = \frac{1}{N} \|\Psi\|_{op}$. Thus both $\tilde{\lambda}_1$ and $\|\Psi\|_1$, as defined in (\ref{eq:L1norm}), are appropriately normalised norms, with values in $[0,1].$ Succinctly put, $\tilde{\lambda}_1$ is a normalised operator norm while $\|\Psi\|_1$ is a normalised $L^1$ norm. For the remainder of this section, our central object of study is the changing value of $\tilde{\lambda}_1$, which represents both the first proportion of explanatory variance, as well as a normalised operator norm of the matrix.

Just like Section \ref{sec:evolutionarycorrelation}, we set $S=90$ and analyse a 90-day rolling window. Let $\tilde{\lambda}_1(t)$ be the normalised first eigenvalue of the matrix $\Psi^{[t-S+1:t]}$. We plot this over our analysis window in Figure \ref{fig:Collective_strength_vs_market_size}. There are two primary findings of interest. First, consistent with earlier experiments, there is a spike in $\tilde{\lambda}_1(t)$ during early 2020. This reflects the highly correlated behaviours of cryptocurrencies and indiscriminate selling during the COVID-19 pandemic. Subsequently, $\tilde{\lambda}_1(t)$ declines until May 2021. This most evident decline corresponds to the cryptocurrency bull market, where total cryptocurrency assets grew by several orders of magnitude. Next, there is a significant rise in $\tilde{\lambda}_1(t)$ towards the end of our analysis window, corresponding to higher collective strength in the market. This increase occurred contemporaneously with the aggressive sell-off in cryptocurrency assets, suggesting that there may be a relationship between the size of the market and the strength of the underlying collective dynamics. Broadly, the trajectory of $\tilde{\lambda}_1(t)$ is highly similar to that of $\nu^\Psi(t)$ as depicted in Figure \ref{fig:Correlation_matrix_norm}.

To investigate further, we quantitatively incorporate the size of the cryptocurrency market changing over time. Let $M_i(t), i=1,...,N, t=0,...,T$ be the multivariate time series of cryptocurrency market sizes $M_i(t)$ over our analysis window. Due to the significant volatility exhibited by the market, we compute a rolling average of the entire market, defined by \\ $\tilde{M}(t) = \frac{1}{S} \sum^t_{k=t-S+1} \sum^N_{i=1} M_i(k), t= S,...,T$. We include the plot of this varying over time in the same figure, Figure \ref{fig:Collective_strength_vs_market_size}.

While previous work \cite{JamescryptoEq} has studied the first eigenvalue in isolation and compared its properties between cryptocurrency and equity markets, this work is the first we know of to examine its relationship with market size over time. In particular, we reveal an inverse relationship between the size of the market and the first eigenvalue of the correlation matrix $\tilde{\lambda}_1(t)$. To quantify this observation, we compute the correlation between the rolling size of the cryptocurrency market, $\tilde{\lambda}_1(t)$, and the collective strength of the market, $\tilde{M}(t)$. The correlation between these two series is computed to be $\rho^{\tilde{M}, \tilde{\lambda}_1} = -0.122 $. While this cannot show causation, it suggests a possibility that as the market grows in size, the strength of collective dynamics may decline. 

This would be a noteworthy finding with important implications for the future of the cryptocurrency market, especially given the divided views on the future of the market's viability. Suppose one of two contrived scenarios exist; a ``bull case'' where cryptocurrency prices recover and cryptocurrency becomes a systemically important asset class, and a ``bear case'' where prices continue to decline and cryptocurrencies lose the interest of institutional investors. Our findings indicate that in the bull case, behaviours may become increasingly fragmented and heterogeneous, and there will be opportunities for skilful security selection to generate portfolio alpha. In the bear case, where prices decline and the size of the market decreases, behaviours may become more homogeneous. This could mean fewer opportunities for alpha generation through security selection, as correlations will be strongly positive. 

%Possible creative interpretation
%This economic behaviour is broadly consistent with what we observe in the equities market. Small caps, or equities with smaller levels of market capitalisation, tend to provide less market efficiency than large caps. It is purported that this is driven by less interest from institutional investors primarily due to lower levels of liquidity. This drives greater heterogeneity in behaviours, and more opportunities for security selection. Our mathematical findings and suggestions for potential future developments in the cryptocurrency market are consistent with the current bifurcation in dynamics between large cap and small cap equities.

\begin{figure*}
    \centering
    \includegraphics[width=0.95\textwidth]{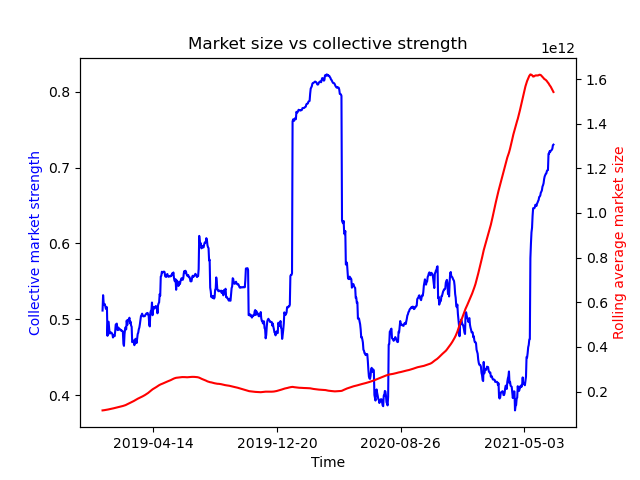}
    \caption{$\tilde{\lambda}_1(t)$, which represents the collective strength of the market and $\tilde{M}(t)$, which represents the size of the cryptocurrency market, as defined in Section \ref{Collective_dynamics_fragmentation}. An inverse relationship is observed with correlation $\rho^{\tilde{M}, \tilde{\lambda}_1} = -0.122$.}
    \label{fig:Collective_strength_vs_market_size}
\end{figure*}

\section{Inconsistency analysis between market size, returns and volatility}
\label{return_vol_size_inconsistency} 

We now extend our study of cryptocurrency market sizes to incorporate their relationship with returns and volatility behaviours individually. Specifically, we investigate the consistency among cryptocurrencies between the attributes of market size, returns and volatility, and how this changes over time. For this purpose, we define three distance matrices and appropriately normalise them.

Let our full period of analysis 01-01-2019 to 30-06-2021 be indexed $t=0,1,...,T$, where $T=911$. Let $M_i(t)$, $t=0,...,T$, be the multivariate time series of market size on each day and let $R_i(t), t=1,...,T$, be the multivariate time series of log returns, as defined by (\ref{eq:logreturns}). Let $S=90$ days and let $\sigma_i(t), t=S,...,T$ be the multivariate time series of 90-day rolling volatility. At each $t$, this is defined as the standard deviation of the log returns of the prior 90 days. Now, we may construct distance matrices for each $t=S,...,T$ as follows:
\begin{align}
\label{eq:marketdiff}
    D^{M}_{ij}(t) = \frac{1}{S}  \left| \sum^t_{k=t-S+1} [M_i(k) - M_j(k)] \right|; \\
    \label{eq:returndiff}
    D^{R}_{ij}(t) =  \left|\sum^t_{k=t-S+1} [R_i(k) - R_j(k)] \right| ; \\
    D^{\Sigma}_{ij}(t) = \left| \sigma_i(t) - \sigma_j(t) \right|.
\end{align}
Thus, $D^M(t), D^R(t)$ and $D^\Sigma (t)$, respectively, measure discrepancy between cryptocurrencies with respect to average market size, total returns, and rolling volatility, each over the 90-day period concluding on day $t$, just like our study of correlation norms and market size in Sections \ref{sec:evolutionarycorrelation} and \ref{Collective_dynamics_fragmentation}. We now convert these three distance matrices into affinity matrices whose elements lie in $[0,1]$:
\begin{align}
    A^{M}_{ij}(t) = 1 - \frac{D_{ij}^{M}(t)}{\max_{kl}\{ D^{M}_{kl}(t) \}}; \\
    A^{R}_{ij}(t) = 1 - \frac{D_{ij}^{R}(t)}{\max_{kl}\{ D^{R}_{kl}(t) \}}; \\
    A^{\Sigma}_{ij}(t) = 1 - \frac{D_{ij}^{\Sigma}(t)}{\max_{kl}\{ D^{\Sigma}_{kl}(t) \}}.
\end{align}
These affinity matrices are appropriately normalised and can be compared directly to study the consistency between cryptocurrency market size and returns or volatility. We generate two \emph{inconsistency matrices} as follows:
\begin{align}
    \text{INC}^{M,R}(t) =A^{M}(t) - A^{R}(t); \\
    \text{INC}^{M,\Sigma}(t) =A^{M}(t) - A^{\Sigma}(t).
\end{align}
Larger absolute values of the entries of $\text{INC}^{M,R}$ indicate that the relationship between two cryptocurrencies regarding market size and returns is quite different, analogously for $ \text{INC}^{M,\Sigma}$. To study the degree of consistency between these attributes in totality across our collection, we compute the $L^1$ norm of the resulting inconsistency matrices and study how these norms evolve over time. That is, for $t=S,...,T$, we compute an analogous quantity as defined in (\ref{eq:L1norm}):
\begin{align}
    \nu_{M,R}^{INC}(t) = \| INC^{M,R}(t) \|; \\
    \nu_{M,\Sigma}^{INC}(t) = \| INC^{M,\Sigma}(t) \|. 
\end{align}

Figure \ref{fig:CCA_time} displays the time-varying inconsistency norms. It is clear that throughout our period of analysis, there is greater consistency between volatility and size than returns and size, as indicated by the smaller values of $\nu_{M,\Sigma}^{INC}(t)$. This is what one would expect in a more established asset class such as equities. For instance, large-cap equities typically exhibit lower volatility than small-cap equities, creating some consistency between market size and volatility. By contrast, returns and size are shown to be significantly more inconsistent, highlighting that the relationship between size and returns is less clear than that between size and volatility. Furthermore, it suggests that the size of cryptocurrencies is by no means a good representation of the future expectation of returns. 

Previous work \cite{James2021_crypto} has used inconsistency matrices to study the extent of inconsistency between returns and volatility. There, we identified the most anomalous individual cryptocurrencies in these attributes. In this work, we take quite a different direction. First, our inconsistency matrices are time-varying. Next, rather than identifying individual cryptocurrencies, we study temporal trends in the collective extent of inconsistency by examining the inconsistency matrices' norms as a function of time. Further, we compare the inconsistency between returns and market size with volatility and market size, incorporating a new parameter. Our finding is also new and of interest, namely that size-returns is more inconsistent than size-volatility. Such a finding could be of great interest to risk managers looking to find factors or exposures their portfolios are most in need of diversifying away from.

\begin{figure*}
    \centering
    \includegraphics[width=0.95\textwidth]{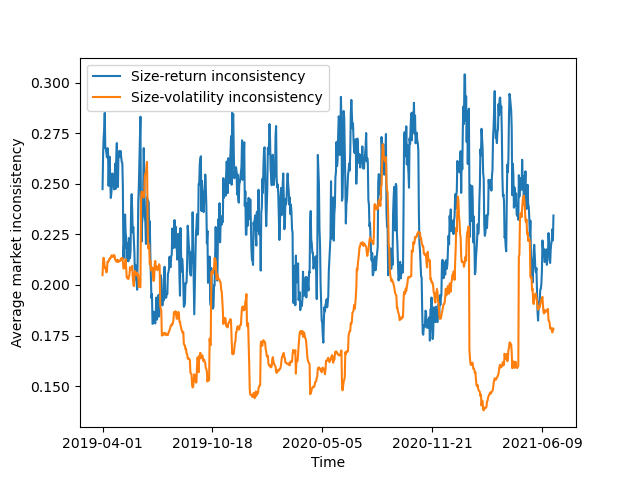}
    \caption{Time-varying inconsistency matrix norms between market size and returns, $\nu_{M,R}^{INC}$, and size and volatility, $\nu_{M,\Sigma}^{INC}$, as defined in Section \ref{return_vol_size_inconsistency}. The size-volatility inconsistency is essentially always lower, indicating more consistency in relationships between cryptocurrencies regarding their size and volatility than their size and returns.}
    \label{fig:CCA_time}
\end{figure*}

\section{Temporal changes and dispersion of volatility}
\label{Volatility_persistence_regime_identification} 

Having identified greater consistency in volatility and market size behaviours, we now more closely examine the structure of collective volatility over time. The behaviour of volatility and the general identification of regimes in financial markets is a topic of great interest. Many parametric statistical models, such as regime-switching models, assume a fixed number of volatility regimes for a candidate modelling problem. Often the selection of this number is quite arbitrary. Like all parametric models, if assumptions are misspecified, the resulting estimates can be highly inaccurate. 

We take a different approach to the analysis of collective volatility behaviours and detect a new phenomenon that we term \emph{volatility dispersion}. To do so, set our window length as $S=90$ days. For each $t=S,...,T$, we may consider the 90-day rolling volatilities $\sigma_i(t)$, as discussed in the previous section. For a fixed $t$, we normalise the vector $(\sigma_1(t),...,\sigma_N(t)) \in \mathbb{R}^N$ by its total sum to produce a probability vector $\mathbf{p}(t)$ that measures the concentration of volatility across a selection of cryptocurrencies. That is, let
\begin{align}
    p_i(t) = \frac{\sigma_i(t)}{\sum_{j=1}^N \sigma_j(t)}.
\end{align}
For example, if $\mathbf{p}(t)=\frac{1}{N}(1,1,...,1) \in \mathbb{R}^N$, this indicates that the 52 cryptocurrencies have identical volatility measured over the past 90-days, while a value of $(1,0,...,0)\in \mathbb{R}^N$ indicates that all volatility is observed in the first currency, with none in any of the others. Thus, $p_i(t)$ is a measure not of the absolute size of volatility but the proportional contribution of one cryptocurrency to the total volatility of the collection.

With these probability vectors $\mathbf{p}(t),t=S,...,T$, we define a $(T-S +1) \times (T-S +1)$ distance matrix using the Wasserstein distance between these distributions at all points in time. That is, let $d^W$ be the $L^1$-Wasserstein metric \cite{DelBarrio}, and let 
\begin{align}
    D^{vol}(s,t) = d^{W} (\mathbf{p}(s), \mathbf{p}(t)) \quad \forall s,t, \in [S,...,T].
\end{align}
We then apply hierarchical clustering to our distance matrix, $D^{vol}(s,t)$, and study the resulting dendrogram Figure \ref{fig:Volatility_dendrogram}. It is worthy to note that the Wasserstein distance, unlike the $L^1$ norm between vectors, does not distinguish probability vectors based on their order; it is essentially a distance between vectors as sets (possibly with repetition). We make this choice because we are not interested in distinguishing periods where a particular cryptocurrency is highly volatile, but whether the volatility is, broadly speaking, spread out or concentrated among the collection as a whole.

\begin{figure*}
    \centering
    \begin{subfigure}[b]{0.91\textwidth}
        \includegraphics[width=\textwidth]{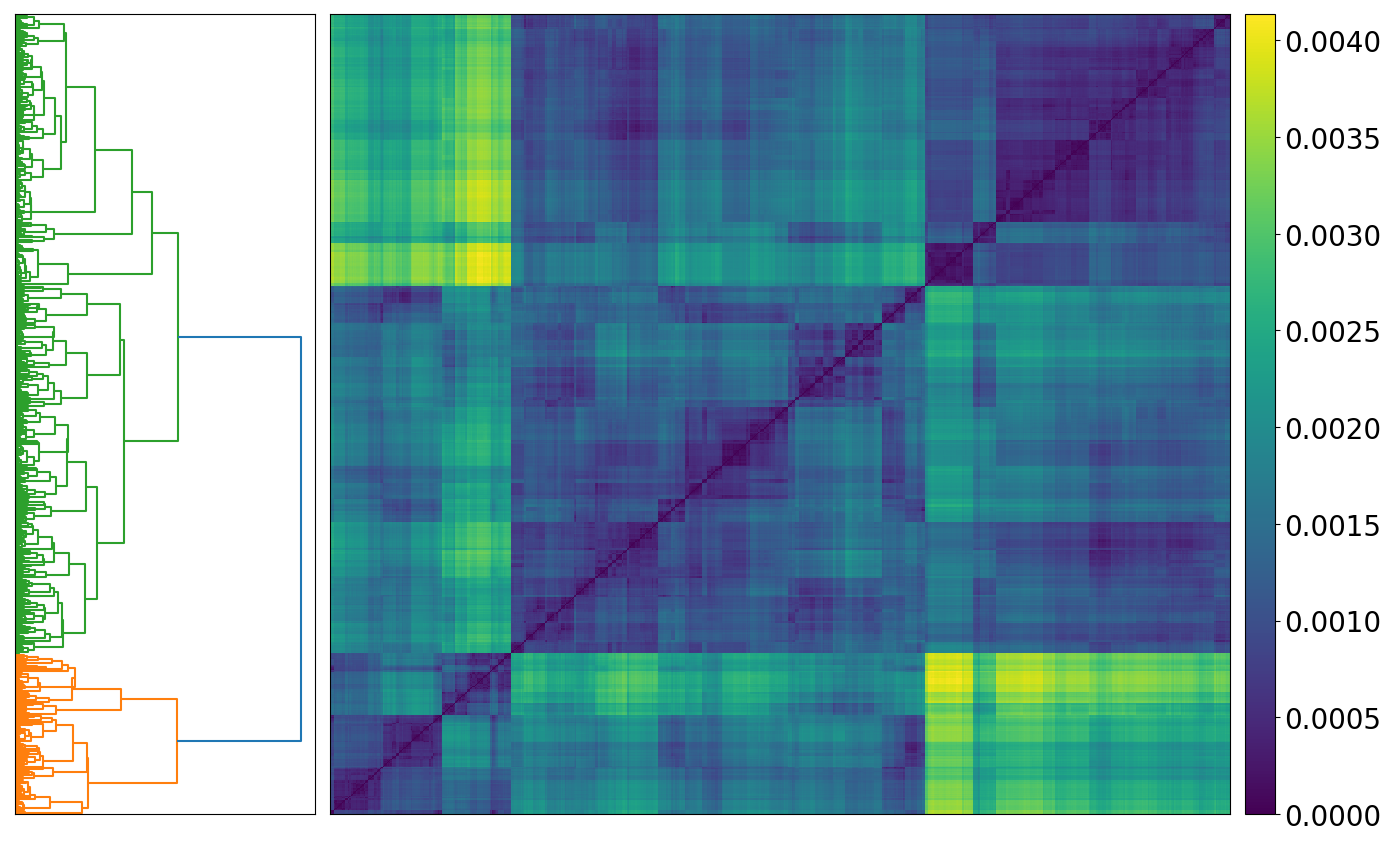}
        \caption{}
   \label{fig:Volatility_dendrogram}
    \end{subfigure}
    \begin{subfigure}[b]{0.91\textwidth}
        \includegraphics[width=\textwidth]{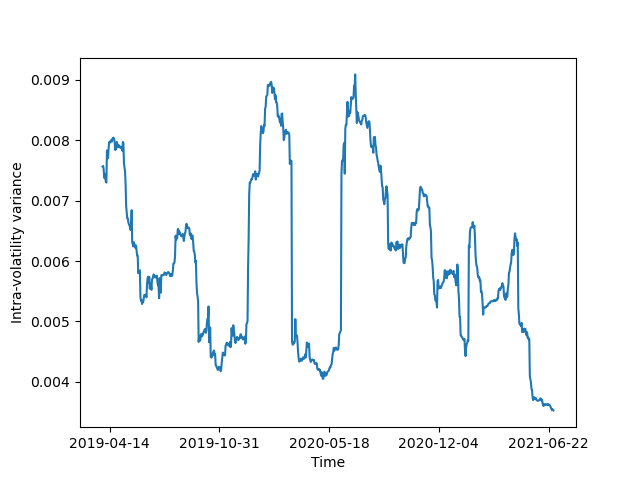}
        \caption{}
   \label{fig:Volatility_variance}
    \end{subfigure}
    \caption{In (a), we perform hierarchical clustering on the matrix $D^{vol}$ between normalised vectors of rolling volatility. The dendrogram groups dates $s,t \in [S,...,T]$ according to similarity between their corresponding vectors $\mathbf{p}(s)$ and $\mathbf{p}(t)$. Two clusters are observed, with the secondary cluster associated to times $t$ observed during COVID-19. In (b), we elucidate this finding more closely, plotting the variance of the probability vector $\mathbf{p}(t)$ over time. The variance of this vector is lower during COVID-19, indicating vectors where all volatilities are closer together. That is, the total volatility of the market is spread out more over all the constituent cryptocurrencies. We term this increased volatility dispersion.}
    \label{fig:Volatility_Wasserstein}
\end{figure*}

Figure \ref{fig:Volatility_dendrogram} highlights several interesting findings. First, two volatility clusters are identified, one dominant cluster of volatility behaviours (with two subclusters) and a smaller cluster that is highly different to the rest of the collection. The latter cluster consists of probability vectors $\mathbf{p}(t)$ generated at times within the COVID-19 market crash. Interestingly, this cluster does not display overwhelming self-similarity; instead, it is relatively diffuse but exhibits significant differences with the concentrated subcluster of the majority cluster. The anomalous behaviours of the COVID-19 market crash are further demonstrated, with such profound distances to other (highly variable) periods in the market over the past several years. 

To further investigate the nature of this split, we perform a closer analysis of the probability vectors $\mathbf{p}(t)$ over time. Considered as a distribution over $[0,1]$, we compute the time-varying \emph{intra-volatility variance} \\ $\text{Var}(\mathbf{p}(t))$. We display this as a function of time in Figure \ref{fig:Volatility_variance}. This supports the separation of behaviours seen in Figure \ref{fig:Volatility_dendrogram}. The COVID-19 market crash exhibits markedly lower variance among the individual rolling volatilities $p_i(t)$. Interestingly, this is also observed towards the end of our analysis window. A lower value of the variance of $\mathbf{p}(t)$ means that the contribution to the variance of each individual cryptocurrency is more uniform across the collection. That is, the proportion of the market's total volatility is more spread out among all the cryptocurrencies during the COVID-19 market crisis. While it is predictable that the absolute volatility would spike during a crisis, it is an unexpected finding that there would be less deviation between the different volatilities of individual cryptocurrencies - essentially, everything is similarly volatile together.

In Appendix \ref{sec:proofs}, we present two theoretical results on the distances $D^{vol}(s,t)$ between normalised volatility vectors and intra-volatility variance $\text{Var}(\mathbf{p}(t))$. These propositions identify the \emph{uniform distribution of volatility} $\mathbf{p}_0=\frac{1}{N}(1,1,...,1) \in \mathbb{R}^N$ and the \emph{one-shot distributions of volatility} $\mathbf{q}_k=(0,...,0,1,0,...,0)$ as the two extremal possible spreads of volatility. The uniform has the lowest intra-volatility variance, a one-shot has the greatest intra-volatility variance, and the greatest possible value of the Wasserstein distance is between the uniform and a one-shot distribution. These propositions demonstrate that, in a precise sense, our study of volatility dispersion investigates the extent that rolling volatility vectors sit between two extremes: the case where all the volatility of the market is uniformly distributed across every asset, and a case where all volatility is concentrated in a single asset.

%IS THE KEY POINT CONVEYED

Fund managers are often warned of the potential issues when over-fitting models to study parameters such as volatility. The volatility dispersion framework we introduce in this section supports the ``Occam's Razor'' principle familiar to many in statistical learning. Rather than trying to capture the true complexity of the volatility process, using just two regimes to capture low and high volatility periods, may work best. However, the interpretations of volatility dispersion are new and different from those gained from traditional volatility clustering, indicating periods where volatility is relatively uniform and less avoidable regardless of asset selection.

\section{Conclusion}
\label{Conclusion}

In Section \ref{Market_correlations}, we study the time-varying evolution of correlations among our collection of cryptocurrencies, and explicitly compare the distribution of correlation coefficients over five discrete time windows. The most notable findings were the spike in correlations during the COVID-19 market crash and the drop in correlations during the subsequent bull market run in late 2020 to early 2021. Broadly, both experiments in this section allude to a clear association between strong collective correlation among cryptocurrencies and periods of declining value in the market. 

In Section \ref{Collective_dynamics_fragmentation}, we investigate the aforementioned association more closely. First, we explore the time-varying explanatory variance of the correlation matrix's first eigenvalue $\tilde{\lambda}_1(t)$, where a noticeable spike is seen during the COVID-19 pandemic. Next, we directly compare $\tilde{\lambda}_1(t)$ with the rolling size of the cryptocurrency market, and show a negative correlation therein. This suggests that as cryptocurrency assets rise, the strength of collective dynamics may weaken. Thus, in a scenario where cryptocurrency assets rise significantly and the asset class gains further prominence, the strength of collective dynamics may decline - leading to more heterogeneous behaviours. This would place greater importance on high-quality security selection when investing in cryptocurrencies.

In Section \ref{return_vol_size_inconsistency}, our experiments reveal greater consistency in volatility and market size behaviours, wherein cryptocurrencies similar in market size are more likely to exhibit commensurate levels of volatility. By contrast, there is less consistency between cryptocurrency size and returns - suggesting that the cryptocurrencies' size do not provide a good indication of future expected returns.

Finally, in Section \ref{Volatility_persistence_regime_identification}, we study the structure of volatility behaviours over time, applying hierarchical clustering to the distances between distributions of rolling volatility at all points in time. Our technique suggests that there are two volatility patterns - times where the total volatility of the market is more dispersed across the entire collection, and times where it is more concentrated in fewer particular cryptocurrencies. We reveal that the COVID-19 market crash not only features higher volatilities in general, but that the total volatility is more evenly spread across all individual cryptocurrencies. This technique could be used as an accompanying tool to estimate the number of regimes in more traditional, parametric regime-switching models in the econometric and statistical modelling literature. Volatility dispersion also provides independent confirmation of and a new approach to studying increased heterogeneity of market dynamics during crises, complementing our study of collective correlations. It is effectively a more direct measure of assets being uniformly volatile together than the first eigenvalue of the correlation matrix.

Many possibilities exist for future work building upon the techniques and findings of this paper. First, one could study the associations between market size and the strength of collective dynamics in the cryptocurrency market with alternative methodologies. For example, suitable distances between normalised trajectories could replace the use of correlations. Second, one could compare studied relationships between the size of the cryptocurrency market and its underlying dynamics with those of more traditional asset classes. Third, given the cryptocurrency market's relative infancy, our findings may turn out to be transient; if the market continues to grow, perhaps the inverse relationship between total market size and collective dynamics will not hold in future crises.

%NEW
Several of our new methodologies and findings may have particular promise in future research and applications. Our time-varying analysis of the total extent of inconsistency between parameters could reveal suitable predictors to incorporate in trading strategies either aimed at maximising returns or minimising volatility. Our volatility dispersion analysis is also promising. Unlike typical methods of volatility clustering or regime-switching models, we compare similarity between all windows of time, not just adjacent periods. In this paper, our analysis has been entirely descriptive, but future work could employ it in a predictive context where we expect volatility to be rather uniform across the market. Such times could prompt an entire withdrawal of funds to safe-haven assets such as gold or cash, as a uniform spread of volatility could mean any investment in the cryptocurrency market would carry significant risk. Finally, volatility dispersion could be applied to other financial and economic securities beyond cryptocurrencies. For instance, one could identify clusters of macroeconomic behaviour using data such as interest rates, GDP, inflation, unemployment, and others. One could explore the dispersion of these factors individually, or identify clusters of economic behaviour with a higher dimensional distance measure where a variety of metrics are incorporated.

%, possibly at specific times deemed similar to past conditions. For example, if time-varying consistency analysis between returns and a different parameter reveals low inconsistency during only bull markets, then this parameter could be incorporated into asset selection specifically only in bull markets. Our identification of volatility dispersion c

Overall, this paper reveals several key relationships between cryptocurrencies' collective dynamics, market size, returns, and volatility and analyses these behaviours over time. During COVID-19 and towards June 30, 2021, correlation behaviours are stronger, and volatility is more uniformly spread across the entire market. Individual cryptocurrencies' market sizes are shown to be more consistent with volatility than returns, while the total market size is inversely associated with the quantifiable strength of collective dynamics. Both the lack of consistency between market size and returns and the high correlations across the market during crises present significant challenges to investors aiming to select optimal portfolios of cryptocurrencies on either a long or short-term basis.

%-----------   END OF PAPER ---------------
\begin{acknowledgements}
The authors would like to thank Georg Gottwald for some helpful discussions. 

\end{acknowledgements}

\section*{Declarations}

\textbf{Funding}: no specific funding was received for this manuscript.

\noindent \textbf{Conflicts of interest}: the authors have no conflicts of interest to report.

\noindent \textbf{Availability of data and material}: all data are publicly available at \url{https://coinmarketcap.com/}

\noindent \textbf{Authors' contributions}: each author is an equal first author, playing an equal role in every aspect of the manuscript.

\appendix

\section{Securities analysed}
\label{appendix:mathematical_objects}

In Table \ref{tab:CryptocurrencyTickers}, we list the 52 cryptocurrencies analysed in this paper, both their tickers and names.

\begin{table*}[ht]
\centering
\begin{tabular}{ |p{1.65cm}|p{3.5cm}|p{1.65cm}|p{3.6cm}|}
% \hline
% \multicolumn{4}{|c|}{\textbf{Cryptocurrency tickers and names}} \\
 \hline
 Ticker & Coin Name & Ticker & Coin Name \\
 \hline
 BTC & Bitcoin & WAVES & WAVES \\
 ETH & Ethereum & CEL & Celsius \\
 BNB & Binance Coin & DASH & Dash  \\
 ADA & Cardano & ZEC & Zcash \\
 XRP & XRP & MANA & Decentraland \\
 DOGE & Dogecoin & ENJ & Enjin Coin \\
 BCH & Bitcoin Cash & HOT & Holo \\
 LTC & Litecoin & QNT & Quant \\
 LINK & Chainlink & KCS & KuCoin \\
 ETC & Ethereum Classic & NEXO & Nexo \\
 XLM & Stellar & BAT & Basic Attention Token \\
 THETA & Theta & ZIL & Zilliqa \\
 VET & VeChain & BTG & Bitcoin Gold\\
 FIL & Filecoin & BNT & Bancor \\
 TRX & Tron & ONT & Ontology \\
 SMR & Monero & ZEN & Horizen \\
 EOS & EOS & SC & Siacoin \\
 CRO & Crypto.com & DGB & Digibyte \\
 MKR & Maker & QTUM & QTUM \\
 BSV & Bitcoin SV & CHSB & SwissBorg \\
 NEO & NEO & ZRX & 0x \\
 XTZ & Tezos & RVN & Ravencoin \\
 MIOTA & IOTA &  OMG &  OMG Network \\ %23 rows
 DCR & Decred &  NANO &  Nano \\ 
 HT & Huobi Token & ICX  & ICON  \\ 
 XEM & NEM & FTM &  Fantom \\ %26 rows
\hline
\end{tabular}
\caption{Cryptocurrency tickers and names}
\label{tab:CryptocurrencyTickers}
\end{table*}

\section{Turning point algorithm}
\label{appendix:TPA}

In this section, we provide more details for the identification of turning points (local maxima and minima) of the smoothed function $\nu_s^\Psi(t)$, used in Section \ref{Market_correlations}. We begin by linearly adjusting this function to ensure that the global minimum of the resulting function is 0. Specifically, let $\hat{\nu}(t)= \nu_s ^{\Psi}(t) - \min_t \{\nu_s^{\Psi}(t)\}, t=S,...,T$.

Following \cite{james2020covidusa}, we apply a two-step algorithm to the smoothed and minimum-adjusted $\hat{\nu}(t)$. The first step produces an alternating sequence of local minima (troughs) and local maxima (peaks), which may include some immaterial turning points. The second step refines this sequence according to chosen conditions and parameters. The most important conditions to initially identify a peak or trough, respectively, are the following:
\begin{align}
\label{baddefnpeak}
\hat{\nu}(t_0)&=\max\{\hat{\nu}(t): \max(0,t_0 - l) \leq t \leq \min(t_0 + l,T)\},\\
\label{baddefntrough}\hat{\nu}(t_0)&=\min\{\hat{\nu}(t): \max(0,t_0 - l) \leq t \leq \min(t_0 + l,T)\},
\end{align}
where $l$ is a parameter to be chosen. We follow \cite{james2020covidusa} in their choice of $l=17$. Defining peaks and troughs according to this definition alone has some flaws, such as the potential for two consecutive peaks.

Instead, we implement an inductive procedure to choose an alternating sequence of peaks and troughs. Suppose $t_0$ is the last determined peak. We search in the period $t>t_0$ for the first of two cases: if we find a time $t_1>t_0$ that satisfies (\ref{baddefntrough}) as well as a non-triviality condition $\hat{\nu}(t_1)<\hat{\nu}(t_0)$, we add $t_1$ to the set of troughs and proceed from there. If we find a time $t_1>t_0$ that satisfies (\ref{baddefnpeak}) and  $\hat{\nu}(t_0)\geq \hat{\nu}(t_1)$, we ignore this lower peak as redundant; if we find a time $t_1>t_0$ that satisfies (\ref{baddefnpeak}) and  $\hat{\nu}(t_1) > \hat{\nu}(t_0)$, we remove the peak $t_0$,  replace it with $t_1$ and continue from $t_1$. A similar process applies from a trough at $t_0$. 

At this point, the time series is assigned an alternating sequence of troughs and peaks. However, some turning points are immaterial and should be removed. Let $t_1<t_3$ be two peaks, necessarily separated by a trough. We select a parameter $\delta=0.2$, and if the \emph{peak ratio}, defined as $\frac{\hat{\nu}(t_3)}{\hat{\nu}(t_1)}<\delta$, we remove the peak $t_3$. If two consecutive troughs $t_2,t_4$ remain, we remove $t_2$ if $\hat{\nu}(t_2)>\hat{\nu}(t_4)$, otherwise remove $t_4$. That is, if the second peak has size less than $\delta$ of the first peak, we remove it.

Finally, we use the same \emph{log-gradient} function between times $t_1<t_2$, defined as
\begin{align}
\label{loggrad}
   \loggrad(t_1,t_2)=\frac{\log \hat{\nu}(t_2) - \log \hat{\nu}(t_1)}{t_2-t_1}.
\end{align}
The numerator equals  $\log(\frac{\hat{\nu}(t_2)}{\hat{\nu}(t_1)})$, a "logarithmic rate of change." Unlike the standard rate of change given by $\frac{\hat{\nu}(t_2)}{\hat{\nu}(t_1)} -1$, the logarithmic change is symmetrically between $(-\infty,\infty)$. Let $t_1,t_2$ be adjacent turning points (one a trough, one a peak). We choose a parameter $\epsilon=0.01$;  if
\begin{align}
    |\loggrad(t_1,t_2)|<\epsilon,
\end{align}
that is, the average logarithmic change is less than 1\%, we remove $t_2$ from our sets of peaks and troughs. If $t_2$ is not the final turning point, we also remove $t_1$. After these refinement steps, we are left with an alternating sequence of non-trivial peaks and troughs - these have been marked on Figure \ref{fig:Correlation_matrix_norm}.

\section{Theoretical properties of volatility dispersion}
\label{sec:proofs}

In the two propositions below, let $\mathbf{p}(t)$ be an arbitrary probability vector of volatilities. That is, $\mathbf{p}(t)=(a_1,...,a_N)\in \mathbb{R}^N$ with $a_i \geq 0$ and $\sum a_i = 1$. Let the \emph{uniform distribution of volatility} be the probability vector given by $\mathbf{p}_0=\frac{1}{N}(1,1,...,1) \in \mathbb{R}^N$. For any $j=1,...,N$, let the \emph{one-shot distribution of volatility} on asset $j$ be $\mathbf{q}_k=(0,...,0,1,0,...,0)$, where the $k$-th coordinate is a 1. That is, $\mathbf{p}_0$ signifies a uniform spread of volatility across all assets, while $\mathbf{q}_k$ represents a spread of volatility exhibited entirely by one asset, with all other assets having zero volatility.

We prove two properties concerning the distance $D^{vol}(s,t) = d^{W} (\mathbf{p}(s), \mathbf{p}(t))$ and the intra-volatility variance $\text{Var}(\mathbf{p}(t))$, as defined in Section \ref{Volatility_persistence_regime_identification}.

\begin{proposition}
Let $\mathbf{p}(t)$ be a volatility vector. The least possible value of $\text{Var}(\mathbf{p}(t))$ is zero, exhibited uniquely by the uniform $\mathbf{p}_0$. The greatest possible value is $1 - \frac{1}{N^2}$, exhibited only by the one-shot distributions $\mathbf{q}_k$.
\end{proposition}
\begin{proof}
The variance can be expressed as
\begin{align}
\text{Var}(\mathbf{p}(t))= \sum_{i=1}^N \left(a_i - \frac{1}{N}\right)^2,
\end{align}
which is clearly at most zero, equality if and only if all $a_i=\frac{1}{N}$, that is, for $\mathbf{p}(t)=\mathbf{p}_0$.

An alternative formulation of the variance is
\begin{align}
\text{Var}(\mathbf{p}(t))&= \sum_{i=1}^N a_i^2 - \frac{1}{N^2} \\
&\leq \sum_{i=1}^N a_i^2 + 2\sum_{i<j} a_{i} a_{j} - \frac{1}{N^2} \\
&= (\sum_{i=1}^N a_i)^2 - \frac{1}{N^2} \\
& = 1 - \frac{1}{N^2}.
\end{align}
Equality is obtained if and only if all all cross terms $a_i a_j=0$ for $i \neq j$. This means that only one value of $a_i$ may be non-zero, say $i=k$. By the condition that $a_i$ sum to 1, this non-zero value must be 1, so $\mathbf{p}(t)= \mathbf{q}_k$.
\end{proof}

\begin{proposition}
The greatest possible distance $D^{vol}(s,t) =\\ d^{W} (\mathbf{p}(s), \mathbf{p}(t))$ occurs between the uniform distribution $\mathbf{p}_0$ and a one-shot distribution $\mathbf{q}_k$. Specifically,
\begin{align}
    D^{vol}(s,t) \leq \frac{2}{N} \left( 1 - \frac{1}{N}\right),
\end{align}
and equality occurs only if $\{\mathbf{p}(t), \mathbf{p}(s) \} = \{ \mathbf{p}_0, \mathbf{q}_k \}$ as sets.
\end{proposition}

\begin{proof}
Let $\mathbf{p}(t)=(a_1,...,a_N)$ be an arbitrary probability vector of volatilities. As discussed in Section \ref{Volatility_persistence_regime_identification}, our implementation of the Wasserstein distance is essentially between vectors treated as sets (possibly with repetition in their elements). Specifically, we associate to each probability vector $\mathbf{p}(t)$ a probability measure defined as a weighted sum of Dirac delta measures
\begin{align}
\label{eq:Wasserstein delta}
    \mu=\frac{1}{N}\sum_{i=1}^N\delta_{a_i}.
\end{align}
When $\mu$ and $\nu$ are probability measures on $\mathbb{R}$ that have cumulative distribution functions $F$ and $G$, there is a simple representation of the Wasserstein distance as follows \cite{DelBarrio}:
\begin{align}
\label{eq:computeWasserstein}
  d^W (\mu,\nu) =  \int_{0}^1 |F^{-1} - G^{-1}| dx,
\end{align}
where $F^{-1}$ and $G^{-1}$ are the quantile functions associated to $F$ and $G$, respectively \cite{Gilchrist2000}. In our scenario, we integrate $\mu$ in (\ref{eq:Wasserstein delta}) to compute $F$. Let $a_{(1)}, a_{(2)},..., a_{(N)}$ be the unique reordering of $a_1,...,a_N$ in non-decreasing order. That is, $a_{(1)} \leq a_{(2)}\leq ...\leq a_{(N)}$ and $\{a_{(1)}, a_{(2)},..., a_{(N)} \} = \{ a_1,...,a_N \}$. Integrating $\mu$, we can see that $F$ is a piecewise-constant increasing step function:
\begin{align}
\label{eq:stepfn}
    F=\sum_{i=1}^{N-1} \frac{i}{n} \mathbbm{1}_{[a_{(i)},a_{(i+1)})} + \mathbbm{1}_{[a_{(N)},\infty)}.%, G= \sum_{j=1}^{n-1} \frac{j}{n}\mathbbm{1}_{[s_j+a,s_{j+1}+a)} + \mathbbm{1}_{[s_n+a,\infty)}.
\end{align}
Thus, its associated quantile function is
\begin{align}
    F^{-1}=\sum_{i=1}^{N} a_{(i)} \mathbbm{1}_{(\frac{i-1}{n},  \frac{i}{n}]}. %G^{-1}=\sum_{j=1}^{n} (s_j+a) \mathbbm{1}_{(\frac{j-1}{n},  \frac{j}{n})}.
\end{align}
Now let $\mathbf{p}(s)=(b_1,...,b_N)$ be an alternative probability vector of volatilities. Let $\nu$, $G$ and $G^{-1}$ be its associated measure, cumulative distribution function and quantile function associated, respectively. In particular,
\begin{align}
    G^{-1}=\sum_{i=1}^{N} b_{(i)} \mathbbm{1}_{(\frac{i-1}{n},  \frac{i}{n}]}. %G^{-1}=\sum_{j=1}^{n} (s_j+a) \mathbbm{1}_{(\frac{j-1}{n},  \frac{j}{n})}.
\end{align}
Thus
\begin{align}
 d^{W} (\mathbf{p}(s), \mathbf{p}(t)) =  \int_{0}^1 |F^{-1} - G^{-1}| dx\\
 = \int_{0}^1 \left| \sum_{i=1}^{N} (a_{(i)} - b_{(i)}) \mathbbm{1}_{(\frac{i-1}{n},  \frac{i}{n}]} \right| dx\\
= \frac{1}{N} \sum_{i=1}^{N} |a_{(i)} - b_{(i)}|\\
= \frac{1}{N} \sum_{i=1}^{N-1} |a_{(i)} - b_{(i)}| + \frac{1}{N}| (a_{(N)}- \frac{1}{N})  - (b_{(N)} - \frac{1}{N}) |\\
\leq \frac{1}{N} \sum_{i=1}^{N-1} (a_{(i)} + b_{(i)}) + \frac{1}{N}( (a_{(N)}- \frac{1}{N})  + (b_{(N)} - \frac{1}{N}) )\\
= \frac{1}{N} \sum_{i=1}^{N} a_{(i)} + \frac{1}{N} \sum_{i=1}^{N} b_{(i)} - \frac{2}{N^2}\\
=\frac{2}{N} - \frac{2}{N^2}\\
= \frac{2}{N}\left(1 - \frac{1}{N}\right).
\end{align}
These inequalities hold by the triangle inequality, and the non-negativity of $a_{(i)}$, $b_{(i)}$, $a_{(N)}- \frac{1}{N}$ and $b_{(N)}- \frac{1}{N}$. These last two are non-negative as the condition $\sum a_i = 1$ implies $a_{(N)}\geq \frac{1}{N}$.

Finally, suppose equality holds. Then for each $i=1,...,N-1$, we must have $|a_{(i)} - b_{(i)}|= a_{(i)} + b_{(i)}$. This means at least one of $a_{(i)}$ or $b_{(i)}$ is zero, for each $i$. Without loss of generality, suppose $a_{(N-1)}=0$. As $ 0 \leq a_{(1)} \leq a_{(2)}\leq ...\leq a_{(N-1)} \leq a_{(N)}$, this immediately shows $a_{(i)}=0$ for all $i=1,...,N-1$. So $a_{(N)}=1$ and up to reordering, $\mathbf{p}(t)$ is a one-shot probability vector $(1,0,...,0)$. Thus, $\mathbf{p}(t)=\mathbf{q}_k$ for some $k$. 

In addition, equality forces $| (a_{(N)}- \frac{1}{N})  - (b_{(N)} - \frac{1}{N}) | = (a_{(N)}- \frac{1}{N})  + (b_{(N)} - \frac{1}{N})$. As both terms are non-negative, this implies one of $a_{(N)}- \frac{1}{N}, b_{(N)}- \frac{1}{N}$ must be zero. But we've established $a_{(N)}- \frac{1}{N}= 1 - \frac{1}{N}$ is non-zero. Thus $b_{(N)}- \frac{1}{N}$ must be zero. So $b_{(N)}=\frac{1}{N}$. And yet $b_{(N)}$ is the largest of the $b_i$, elements that sum to 1. This implies all the $b_i$ must be equal to $\frac{1}{N}$. So $\mathbf{p}(s)$ must be the uniform distribution $\mathbf{p}_0$. It is then routine to demonstrate that $d^W(\mathbf{q}_k, \mathbf{p}_0)$ attains the theoretical upper bound of $\frac{2}{N}\left(1 - \frac{1}{N}\right)$.
\end{proof}

\begin{remark}
The above proof can be made more conceptual by noticing that the condition $\sum_{i=1}^N a_i = 1$ shows directly that
\begin{align}
    \int_{0}^1 F^{-1}  dx = \frac{1}{N}.
\end{align}
Thus, there is a quick bound on $ \int_{0}^1 |F^{-1} - G^{-1}| dx$ of $\frac{2}{N}$. To tighten the bound, one can notice that the quantile functions $F^{-1}$ and  $G^{-1}$ necessarily have a ``rectangle in common'' in their area under the curve. This rectangle has area $\frac{1}{N^2}$, which provides the strengthening of the bound. If $x$ is the coordinate on the domain of $F^{-1}$ and $G^{-1}$ and $y$ is the coordinate on the codomain, this rectangle can be described as $ 1 - \frac{1}{N} \leq x \leq 1, 0 \leq y \leq \frac{1}{N}$ (so is in fact a square).
\end{remark}

Viewed together, these propositions demonstrate that, in a precise sense, our study of volatility dispersion investigates the extent that real rolling volatility data sits between two extremes: the case where all the volatility of the market is uniformly distributed across every asset, and a case where all volatility is concentrated in a single asset.

%the fact that $a_{(i)}$ and $b_{(i)}$ are non-negative, and the fact that $a_{(N)}- \frac{1}{N}$ and $b_{(N)}- \frac{1}{N}$ is necessarily non-negative.

%\bibliographystyle{spbasic}% basic style, author-year citations
%\bibliographystyle{spmpsci} % mathematics and physical sciences
\bibliographystyle{spphys} % APS-like style for physics
\bibliography{_References} 
\end{document}